\documentclass{amsart}

\title
{KP solitons and the Riemann theta functions}

\author{Yuji Kodama$^{*\dagger}$} 
\date{\today}
\address{$^*$ College of Mathematics and Systems Sciences, Shandong University of Science and Technology, Qingdao, 266590, China}
\address{
$^\dagger$ Department of Mathematics, The Ohio State University, Columbus, OH 43210, USA}
\email{kodama@math.ohio-state.edu}
\subjclass[2000]{}

\usepackage{url}
\usepackage{color}
\usepackage{rotating}

\countdef\x=23
\countdef\y=24
\countdef\z=25
\countdef\t=26

\def\tbox(#1,#2)#3{
\x=#1 \y=#2 
\multiply\x by 12 
\multiply\y by 12 
\z=\x \t=\y
\advance\z by 12 
\advance\t by 12 
\psline(\x,\y)(\x,\t)(\z,\t)(\z,\y)(\x,\y)
\advance\x by 6
\advance\y by 6 
\rput(\x,\y){{\bf #3}}}

\usepackage{amssymb}
\usepackage{graphicx}
\usepackage{epstopdf}
\usepackage{amscd}
\usepackage{appendix}
\usepackage{graphics}
\usepackage{tikz}
\def\proof{\par{\it Proof}. \ignorespaces}
\def\endproof{{\ \vbox{\hrule\hbox{%
     \vrule height1.3ex\hskip0.8ex\vrule}\hrule }}\par}

\theoremstyle{definition}

\theoremstyle{remark}

\numberwithin{equation}{section}

\let\trueint=\int
\let\truesum=\sum
\def\int{\mathop{\textstyle\trueint}\limits}
\def\sum{\mathop{\textstyle\truesum}\limits}
\let\<=\langle
\let\>=\rangle

\def\sech{\mathop{\rm sech}\nolimits}

\def\t{\mathbf{t}}
\def\0{\mathbf{0}}
\def\z{\mathbf{z}}

\def\edge{\ar@{-}}
\def\dedge{\ar@{.}}

\usepackage{amsmath, amsthm, amssymb, amsbsy,amsfonts,amscd}
\usepackage[mathscr]{eucal}
\usepackage{epsfig}
\usepackage{young}
\newtheorem{theorem}{Theorem}[section]

\newtheorem{definition}[theorem]{Definition}
\newtheorem{proposition}[theorem]{Proposition}
\newtheorem{lemma}[theorem]{Lemma}

\newtheorem{example}[theorem]{Example}

\newtheorem{remark}[theorem]{Remark}

\usepackage{amsfonts}
\usepackage{xy}
\usepackage{amssymb}
\usepackage{amsmath}

\newcommand{\C}{\mathbb C}

\newcommand{\thmrefer}[1]{\renewcommand\thetheorem
  {\protect\ref{#1}}\addtocounter{theorem}{-1}}

\xyoption{all}
\CompileMatrices

\begin{document}

\begin{abstract}
We show that the $\tau$-functions of the regular KP solitons from the totally nonnegative Grassmannians can be expressed by the Riemann theta functions on singular curves.  We explicitly write the parameters in the Riemann theta function in terms of those of the KP soliton. We give a short remark on the Prym theta function on a double covering of singular curves.
We also discuss the KP soliton on quasi-periodic background, which is obtained by applying the vertex operators to the Riemann theta function.
\end{abstract}

\maketitle

\tableofcontents


\section{Introduction}
It is well known that the KP solitons are given by the so-called $\tau$-function in the form of Wronskian determinant (see e.g. \cite{H:04} and below).
The regular KP solitons have been classified in terms of the totally non-negative Grassmannian $\text{Gr}(N,M)$ where $N$ and $M$ are positive integers with $N<M$ \cite{CK:09, KW:13}.
These solitons are uniquely determined by a set of distinct parameters $(\kappa_1,\kappa_2,\ldots,\kappa_{M})\in\mathbb{R}^M$ and an element $A\in\text{Gr}(N,M)$  (see e.g. \cite{KW:14, K:17}). It is also known that the quasi-periodic (algebro-geometric) solutions of the KP equation can be constructed by the Riemann theta function on smooth curve (see e.g. \cite{BBEIM:94}).The real regular solutions are also discussed in \cite{BBEIM:94}.

In this paper, we show that the regular KP solutions can be expressed by the theta function on singular curve with several degenerations of double points. The theta function on singular curve with ordinary double points was 
obtained by Mumford in his famous book \cite{Mu:84}. We call this theta function the $M$-theta function in this paper. Then we give explicit formulas of the $M$-theta function in terms of $(\kappa_1,\ldots,\kappa_M)$ and $A\in\text{Gr}(N,M)$. It turns out that a general KP soliton is obtained by a further degeneration of the nodal points. It is quite interesting to study these singular structure from a viewpoint of algebraic geometry. But we will not discuss the details in this note. We would like to mention a recent paper \cite{I:22}, where the singular structure of the Riemann surface has been studied in connection with the $\tau$-function of the KP solutions \cite{I:22}. We also give a remark on the Prym theta function as a $\tau$-function for the BKP hierarchy, which is given by a Pfaffian form \cite{Mu:74,H:89,Ta:91}.

The paper is organized as follows. In Section \ref{sec:Riemann}, we give a brief description of the Riemann theta function on singular curve following Chapter 5 in the book \cite{Mu:84}. We define the $M$-theta function and note that the $M$-theta function can be expressed in the form with vertex operators. In Section \ref{sec:KP}, we start with a brief introduction of the regular soliton solutions of the KP equation. The section will be concluded with the main theorem (Theorem \ref{thm:main}),
which gives explicit formula of the $\tau$-function in terms of the $M$-theta function. We consider two main examples in Section \ref{sec:Example}, where we explicitly demonstrate how one can construct the regular KP soliton directly from the matrix $A\in \text{Gr}(N,M)$.  In the final section \ref{sec:SolitonQ},
we give a formula which describes the KP solitons on quasi-periodic background. The formula is obtained by applying the vertex operators given in Section \ref{sec:Riemann}.


\section{The Riemann theta function}\label{sec:Riemann}
The Riemann theta function of genus $g$ is defined as follows.
Let $\mathcal{H}_g$ be the set of symmetric $g\times g$ matrix whose imaginary part is positive definite, called the Siegel upper-half-plane and the element $\Omega\in\mathcal{H}_g$ is called the period matrix.  The Riemann theta function of genus $g$ is then defined by
\begin{equation}\label{eq:Riemann}
\vartheta_g({\bf z};\Omega):=\sum_{{\bf m}\in\mathbb{Z}^g}\exp 2\pi i \left(\frac{1}{2}{\bf m}^T\Omega {\bf m}+{\bf m}^T {\bf z}\right),
\end{equation}
for ${\bf z}\in \C^g$, and $\mathbf{m}^T$ is the transpose of the column vector $\mathbf{m}$.	The $g\times g$ period matrix $\Omega$ is calculated as follows (see e.g. \cite{Mu:83, FK:91}). Let $\mathcal{C}$ be a compact Riemann surface of genus $g$. Let $\{a_1,\ldots,a_g, b_1,\ldots,b_g\}$ be a canonical homology basis on $\mathcal{C}$. Let $\{\omega_1,\ldots,\omega_g\}$ be a basis of the normalized holomorphic differentials on $\mathcal{C}$. Then the period matrix $\Omega=(\Omega_{i,j})$ is calculated by
\[
\oint_{a_i}\omega_j=\delta_{i,j},\qquad \oint_{b_i}\omega_j=\Omega_{i,j},\qquad \text{for}\quad 1\le i,j\le g.
\]
For example of the case of $g=1$, we consider an elliptic curve $\mu^2+\prod_{j=1}^3(\lambda-e_j)=0$ with the real parameters $e_1>e_2>e_3$. Then we have 
\[
\omega=\frac{d\lambda}{A\mu} \quad\text{which gives}\quad \Omega={\oint_b\omega}=i\frac{K'(k)}{K(k)},
\]
where $A=\oint_a\frac{d\lambda}{\mu}$, $K'(k)=K(\sqrt{1-k^2})$, and  $K(k)$ is the complete elliptic integrals of the first kind,
\[
K(k)=\int_0^1\frac{dt}{\sqrt{(1-t^2)(1-k^2t^2)}},\quad\text{with}\quad k^2=\frac{e_2-e_3}{e_1-e_3}.
\]
The equation \eqref{eq:Riemann} is the Jacobi theta function.

We just mention that in general, an explicit computation of the period matrix is hard for given algebraic curve.
(See e.g. \cite{G:89, Mu:83, Mu:84}.)

\subsection{The Riemann theta function on a singular curve}
In \cite{Mu:84}[Chapter 5, p.3.243], Mumford considered the theta function on singular curve. Let $C$ be a singular curve of (arithmetic) genus $g$, and let $S=\text{Sing}(C)=\{p_1,\ldots,p_g\}\subset C$.  Assume that the singularities of $C$ are only ordinary double points $p_1,\ldots,p_g$ and that $C$ has normalization
\[
\pi:\mathbb{P}~\to~ C.
\]
If $\pi^{-1}(p_i)=\{b_i,c_i\}$ for $i=1,\ldots,g$, this implies that $C$ is just $\mathbb{P}$ with $g$ pairs of points $\{b_i,c_i\}$ identified. Then the Riemann theta function on the singular curve $C$ can be obtained by taking the limits of the diagonal parts of the period matrix $\Omega$,
\begin{equation}\label{eq:limits}
\text{Im}~\Omega_{i,i}~\to~ \infty \quad \text{for}\quad 1\le i\le g,
\end{equation}
and the limits of $\Omega_{k,l}$ for $k\ne l$ are given by
\begin{equation}\label{eq:Omega}
\Omega_{k,l}~\to~\tilde\Omega_{k,l}:=\frac{1}{2\pi i}\ln\frac{(b_k-b_l)(c_k-c_l)}{(b_k-c_l)(c_k-b_l)},
\end{equation}
(see also \cite{F:73, Ka:11}).
The limit of \eqref{eq:Riemann} will be just $1$, which corresponds to the choice ${\bf m}^T=(0,\ldots,0)$. To obtain a nontrivial example, we consider the shifts
\[
z_i~\to~z_i-\frac{1}{2}\Omega_{i,i},\quad\text{for}\quad i=1,\ldots,g,
\]
which then gives the Riemann theta function with shifted variable ${\bf z}=(z_1,\ldots,z_g)^T\in\C^g$,
\begin{equation}\label{eq:RiemannS}
\vartheta_g({\bf z};\Omega)=\sum_{{\bf m}\in\mathbb{Z}^g}\exp 2\pi i\left(\frac{1}{2}\sum_{i=1}^gm_i(m_i-1)\Omega_{i,i}+\sum_{i<j}m_im_j\Omega_{i,j}+\sum_{i=1}^gm_iz_i\right),
\end{equation}
where ${\bf m}=(m_1,\ldots,m_g)$.
Then the limits $\Omega_{i,i}\to+i\infty$ for all $i=1,\ldots, g$ lead to
\begin{align}\label{eq:Theta}
\vartheta_g(\z;\Omega)~\to~ &~\tilde\vartheta^{(g)}_g(\z;\tilde\Omega):=\sum_{{\bf m}\in\{0,1\}^g}\exp 2\pi i\left(\sum_{k<l}m_km_l\tilde\Omega_{k,l}+\sum_{k=1}^gm_kz_k\right)\\
=&1+\sum_{k=1}^g e^{2\pi i z_k}+\sum_{1\le k<l\le g}e^{2\pi i\tilde\Omega_{k,l}}e^{2\pi i (z_k+ z_l)}+\ldots+e^{2\pi i\sum_{k<l}\tilde\Omega_{k,l}}e^{2\pi i\sum_{j=1}^g z_j},\nonumber
\end{align}
that is, the {infinite} sum of exponential terms in the $\vartheta$-function \eqref{eq:Riemann} becomes a finite sum of $2^g$ exponential terms corresponding to the case with $m_i\in\{0,1\}$.
We remark here that $\tilde\vartheta(\z;\tilde\Omega)$ in \eqref{eq:Theta} can be expressed in a simple form using the vertex operator $\hat V_j[\tilde\Omega]$, which is defined by
\begin{equation}\label{eq:VO}
\hat{V}_j[\tilde\Omega]:=\exp 2\pi i\left(z_j+\sum_{k=1, k\ne j}^g\tilde\Omega_{j,k}\frac{\partial}{\partial z_k}\right).
\end{equation}
It is easy to verify the following lemma for the product of two vertex operators.
\begin{lemma}\label{lem:VO}
For $j\ne k$, we have
\begin{align}\label{eq:VOR}
\hat V_j[\tilde\Omega]\cdot \hat V_k[\tilde\Omega]=\frac{(b_j-b_k)(c_j-c_k)}{(b_j-c_k)(c_j-b_k)}:\hat V_j[\tilde\Omega]\cdot\hat V_k[\tilde\Omega]:,\nonumber
\end{align}
where $:\hat V_j\cdot\hat V_k:$ is the normal ordering of the product $\hat V_j\cdot\hat V_k$, and we have used \eqref{eq:Omega}. 
\end{lemma}
Lemma \ref{lem:VO} implies that the vertex operators commute.
Then the function $\tilde\vartheta(\z;\tilde\Omega)$ in \eqref{eq:Theta} is expressed by
\begin{equation}\label{eq:Ltheta}
\tilde\vartheta^{(g)}_g(\z;\tilde\Omega)=\prod_{j=1}^g(1+\hat V_j[\tilde\Omega])\cdot 1.
\end{equation}

We call the $\tilde\vartheta^{(g)}_g$ function the $M$-theta function after Mumford (see Chapter 5 in his book \cite{Mu:84}).
We will also define $\tilde\vartheta_g^{(n)}$, which is a theta function on a singular curve of genus $g-n$ having with $n$ singular points (see Section \ref{sec:SolitonQ}).
In \cite{Mu:84}, the $g$ pairs of points $\{b_i,c_i\}$ on $\mathbb{P}$ are taken as $b_i=-c_i$ to show that  \eqref{eq:Ltheta} gives a $g$-soliton solution of the KdV equation with appropriate choice of the variables $\z$. One can also see that the $M$-theta function gives $g$-soliton solution of the KP equation (see below).
In this paper, we show that any KP soliton generated by the totally nonnegative Grassmannian (see e.g. \cite{K:17}) can be written in the form \eqref{eq:Ltheta} with appropriate restrictions on $z_j$ and $\tilde\Omega_{k,l}$, i.e. further limits of $\Omega_{k,l}$, which is also related to an identification with points $\{b_i,c_i\}$ on $\mathbb{P}$. 

\begin{remark}\label{rem:Grammian}
It is also interesting to note that the $M$-theta function can be expressed in the determinant form \cite{Ka:11},
\[
\tilde\vartheta_g^{(g)}(\z;\tilde\Omega)=\text{det}\left(\delta_{j,k}+\frac{b_j-c_j}{b_j-c_k}e^{\pi i(z_j+z_k)}\right)_{1\le j,k \le g}
\]
Expressing the variables $\z$ and $\tilde\Omega$ in terms of the soliton parameters (see below), this formula can be identified as the Grammian form \cite{CLM:10, H:04}.
\end{remark}

\subsection{The Prym theta function on a singular curve}
Here we give a remark on the theta function of Prym varieties of branched coverings \cite{Mu:74, DJKM:82,Ta:91}.
First, we note that the period matrix $\tilde\Omega_{j,k}$ in \eqref{eq:Omega} is obtained by the integral,
\[
\tilde\Omega_{j,k}=\frac{1}{2\pi i}\int_{c_j}^{b_j}\,\omega_k(z),
\]
where $\omega_k$ is the differential of third kind (which is obtained by the holomorphic differential of the first kind in the singular limit \cite{Ka:11, I:22}),
\[
\omega_k(z)=\left(\frac{1}{z-b_k}-\frac{1}{z-c_k}\right)\,dz,\qquad \text{for}\quad k=1,\ldots,g.
\]

In the theory of Prym variety \cite{Mu:74}, we have a double covering map, $\mu:\tilde{C}\to C$ for a curve $C$ of genus $g$.
We also have an involution, interchanging two sheets above any point in $C$, $\iota:\tilde{C}\to \tilde C$ (see also \cite{Ta:91}).
The involution acts on the differentials $\omega_k^B(z)$ by
\[
\iota(\omega_k^B(z))=\omega^B_k(-z)=-\omega_k^B(z).
\]
Then the normalization $\tilde\pi:\mathbb{P}\to\tilde{C}$ has the inverse image $\tilde\pi^{-1}(p^{\pm}_j)=\{\pm b_j,\pm c_j\}$ with $\mu(p^{\mp}_j)=p_j\in C$.
The differentials for a singular curve $C$ may be expressed by
\[
\omega_k^B(z)=\left(\frac{1}{z-b_k}-\frac{1}{z-c_k}-\frac{1}{z+b_k}+\frac{1}{z+c_k}\right)\,dz.
\]
Note that one can write $\omega_k^B=\omega^+_k+\omega_k^-$, where
\[
\omega^{\pm}_k(z)=\pm\left(\frac{1}{z\mp b_k}-\frac{1}{z\mp c_k}\right)\,dz.
\]
That is, each $\omega_k^{\pm}$ is the differential on one sheet of double covers $\tilde{C}$.
Then the period matrix $\tilde\Omega^B_{k,l}$ can be computed as
\[
\tilde\Omega^B_{j,k}=\frac{1}{2\pi i}\int_{c_j}^{b_j}\,\omega_k^B(z)=\frac{1}{2\pi i}\ln \frac{(b_j-b_k)(c_j-c_k)(b_j+c_k)(c_j+b_k)}{(b_j-c_k)(c_j-b_k)(b_j+b_k)(c_j+c_k)}.
\]
We now define the Prym $M$-theta function corresponding to two-sheeted coverings of singular curve, which is defined  in the same form of \eqref{eq:Theta} with $\tilde\Omega_{j,k}=\tilde\Omega_{j,k}^B$,
\begin{equation}\label{eq:Btheta}
\tilde\vartheta^B_g(z;\tilde\Omega^B):=\sum_{{\bf m}\in\{0,1\}^g}\exp 2\pi i\left(\sum_{j<k}m_jm_k\,\tilde\Omega^B_{j,k}+\sum_{k=1}^gm_kz_k\right).
\end{equation}

\begin{remark}\label{rem:Pfaffian}
Introduce new variables $\{\xi_j:1\le j\le 2g\}$ and write
\[
2\pi i z_k=\xi_{2k-1}+\xi_{2k}+\ln\frac{b_k+c_k}{b_k-c_k},\qquad \text{for}\quad k=1,\ldots,g.
\]
Then, following \cite{H:89}, one can show that the Prym $M$-theta function can be expressed by the Pfaffian of $2g\times 2g$ skew symmetric matrix,
\begin{equation}\label{eq:Btau}
\tilde\vartheta_g^B(z;\tilde\Omega^B)=\text{Pf}\left(c_{j,k}+\frac{q_j-q_k}{q_j+q_k}\,e^{\xi_j+\xi_k}\right)_{1\le j<k\le 2g},
\end{equation}
where $c_{j,k}=-c_{k,j}$ are constants, and for $1\le j<k\le 2g$ and $1\le l\le 2g$,
\begin{align*}
c_{j,k}&=\left\{\begin{array}{llll}
1,\quad&\text{if}~\, j=2n-1~\text{and}~k=2n,\quad(1\le n\le g),\\
0,\quad&\text{otherwise}
\end{array}\right.\\
q_l&=\left\{\begin{array}{lll}
b_n,\quad&\text{if}~\, l=2n-1,\\
-c_n,\quad&\text{if}~\, l=2n,
\end{array}\right.\quad\text{for}\quad1\le n\le g.
\end{align*}
Notice that the Prym $M$-theta function can be also expressed using the vertex operators in the same way as \eqref{eq:Ltheta} with $\tilde\Omega=\tilde\Omega^B$.
We will also note the Pfaffian formulas of the soliton solutions for the BKP hierarchy in the next section.
\end{remark}


\section{The KP equation}\label{sec:KP}
In this section, we give a basic information about the KP solitons for the purpose of the present paper (see e.g. \cite{K:17} for the details). The KP equation is a nonlinear partial differential equation in the form,
\begin{equation}\label{eq:KP}
\partial_x(-4\partial_t u+6u\partial_xu+\partial_x^3u)+3\partial_y^2u=0,
\end{equation}
The solution of the KP equation is given in the following form,
\begin{equation}\label{eq:u}
u(x,y,t)=2\partial_x^2\ln \tau(x,y,t),
\end{equation}
where $\tau(x,y,t)$ is called the $\tau$-function of the KP equation.

The $\tau$-function satisfies the Hirota bilinear equation,
\[
(-4D_xD_t+D_x^4+3D_y^2)\tau\circ\tau=0,
\]
where $D_x$'s are the Hirota derivatives defined by
\[
D_x^m f\circ g:=(\partial_x-\partial_{x'})^m f(x)g(x')\Big|_{x=x'}
\]

\subsection{Quasi-periodic solutions}
A quasi-periodic solution is given by a Riemann theta function
\[
\tau(x,y,t)=\vartheta_g({\bf{u}}x+{\bf v}y+{\bf w}t+{\bf z}_0,\Omega),
\]
where $\bf{u},\bf{v}, \bf{w}$ and ${\bf z}_0$ are real $g$-dimensional constant vectors. The vectors ${\bf u},{\bf v}$ and ${\bf w}$ satisfy the nonlinear dispersion relation (more precisely, the Schottky problem). The constant vectors are given by integrals of normalized meromorphic differentials (see e.g. \cite{BBEIM:94}).

In the case of $g=1$, we take
\[
z=\frac{1}{4K(k)}(\alpha x+\beta y+\gamma t+\omega_0)+\frac{1}{2},
\]
where the nonlinear dispersion for the constants $(\alpha, \beta, \gamma)$ is
\[
\frac{-4\alpha\gamma+3\beta^2}{\alpha^4}=3\frac{E(k)}{K(k)}-2+k^2.
\]
Here $E(k)$ is the complete elliptic integral of the second kind,
\[
 E(k)=\int_0^1\frac{\sqrt{1-k^2t^2}}{\sqrt{1-t^2}}\,dt,\qquad\text{for}\quad 0\le k\le 1.
\]
Then the solution $u(x,y,t)$ is given by
\begin{equation}\label{eq:dn}
u(x,y,t)=2\partial_x^2\log \vartheta_1(z;\Omega)=2\alpha^2\left(\text{dn}^2\frac{\alpha x+\beta y+\gamma t+\omega_0}{2}-\frac{E(k)}{K(k)}\right).
\end{equation}
where $\text{dn}\,z$ is the Jacobi elliptic function. 
Then we take the limit $k\to 1$, which is called the soliton limit, and we have
\[
u(x,y,t)~\to~2\alpha^2\sech^2\frac{\alpha x+\beta y+\gamma t+\omega_0}{2}.
\]
Here we have used $K(k)\to\infty$ and $K'(k)\to\frac{\pi}{2}$, and note
\[
\text{dn}^2 z-\frac{E}{K}=\left(\frac{\pi}{2K'}\right)^2\sum_{l\in\mathbb{Z}}\sech^2\frac{\pi K}{K'}\left(\frac{z}{2K}-l\right)~\to~\sech^2 z.
\]
Note also that the dispersion relation for the soliton solution is given by
\begin{equation}\label{eq:dispersion}
-4\alpha\gamma+3\beta^2+\alpha^4=0
\end{equation}
\begin{remark}
In this paper, we consider the limit $k\to 0$, which gives $K'\to\infty$, i.e.
\[
\Omega~\to~+i\infty,\quad \text{and}\quad K(k),~E(k)~\to~\frac{\pi}{2}.
\]
This limit corresponds to the linear limit. However, considering $z$ to be imaginary and shifting $z$ by $\frac{\Omega}{2}$, one can also obtain the soliton solution as will be explained below. Note that changing the wave parameters $\kappa:=(\alpha,\beta, \gamma)$ to $i\kappa$, the dispersion relation becomes \eqref{eq:dispersion}.
\end{remark}

\subsection{Soliton solutions}
The soliton solutions are constructed as follows:
Let $\{f_i(x,y,t):1\le i\le N\}$ be a set of linearly independent functions $f_i(x,y,t)$  satisfying the following system of linear equations,
\begin{equation}\label{eq:f}
\frac{\partial f_i}{\partial y}=\frac{\partial^2 f_i}{\partial x^2},\quad \text{and}\quad  \frac{\partial f_i}{\partial t}=\frac{\partial^3 f_i}{\partial x^3}\qquad i=1,\ldots,N.
\end{equation}
The $\tau$-function is then given by the Wronskian with respect $x$-variable,
\begin{equation}\label{eq:Wr}
\tau(x,y,t)=\text{Wr}(f_1,f_2,\ldots,f_N).
\end{equation}
(See, e.g. \cite{K:17} for the details.)

As a fundamental set of the solutions of \eqref{eq:f}, we take the exponential functions,
\begin{equation*}\label{eq:E}
E_j(x,y,t)=e^{\xi_j(x,y,t)}\quad\text{with}\quad \xi_j(x,y,t):=\kappa_jx+\kappa_j^2y+\kappa_j^3t\qquad \text{for}\quad j=1,\ldots,M.
\end{equation*}
where $\kappa_j$'s are arbitrary real constants. In this paper, we consider the regular soliton solutions, for which we assume the ordering
\begin{equation}\label{eq:orderK}
\kappa_1\,<\,\kappa_2\,<\,\kappa_3\,<\,\cdots\,<\,\kappa_M.
\end{equation}
For the soliton solutions, we consider $f_i(x,y,t)$ as a linear combination of the exponential solutions,
\begin{equation}\label{eq:Ef}
f_i(x,y,t)=\sum_{j=1}^Ma_{i,j}E_j(x,y,t)\qquad\text{for}\qquad i=1,\ldots,N.
\end{equation}
where $A:=(a_{i,j})$ is an $N\times M$ constant matrix of full rank, $\text{rank}(A)=N$.
Then the $\tau$-function \eqref{eq:Wr} is expressed by
\begin{equation}\label{eq:tauE}
\tau(x,y,t)=|AE(x,y,t)^T|,
\end{equation}
where $E(x,y,t)^T$ is the transpose of the matrix $E(x,y,t)$, which is given by
\begin{equation}\label{eq:Et}
E(x,y,t)=\begin{pmatrix}
E_1 & E_2 &\cdots & E_M\\
\partial_1E_1&\partial_1E_2&\cdots &\partial_1E_M\\
\vdots &\vdots &\ddots &\vdots\\
\partial_1^{N-1}E_1 &\partial_1^{N-1}E_2&\cdots &\partial_1^{N-1}E_M
\end{pmatrix}.
\end{equation}
Note here that the set of exponential functions $\{E_1,\ldots,E_M\}$ gives a basis of $M$-dimensional space
of the null space of the operator $\prod_{i=1}^M(\partial_x-\kappa_i)$, and we call it a ``basis''
of the KP soliton.  Then the set of functions $\{f_1,\ldots,f_N\}$ represents an $N$-dimensional subspace
of $M$-dimensional space spanned by the exponential functions. This leads naturally to the structure
of a finite Grassmannian $\text{Gr}(N,M)$, the set of $N$-dimensional subspaces in $\C^M$.
Then the $N\times M$ matrix $A$ of full rank can be identified as a point of $\text{Gr}(N,M)$, and in this paper, we assume
$A$ to be in the reduced row echelon form (RREF). Here we consider the irreducible matrix $A$, which is defined as follows.
\begin{definition}
An $N\times M$ matrix $A$ in RREF is irreducible, if
\begin{itemize}
\item[(a)] in each row, there is at least one nonzero element besides the pivot, and
\item[(b)]  there is no zero column.
\end{itemize}
This implies that the first pivot is located at $(1,1)$ entry, and the last pivot should be at $(N, i_N)$ with $N\le i_N<M$.
\end{definition}

The $\tau$-function can be expressed as the following formula using the Binet-Cauchy lemma (see e.g. \cite{K:17}),
\begin{equation}\label{eq:tauPl}
\tau(x,y,t)=\sum_{I\in\binom{[M]}{N}}\Delta_{I}(A)E_{I}(x,y,t),
\end{equation}
where $I=\{i_1<i_2<\cdots<i_N\}$ is an $N$ element subset in $[M]=\{1,2,\ldots,M\}$, $\Delta_{I}(A)$ is the $N\times N$ minor with the column vectors indexed by $I=\{i_1,\ldots,i_N\}$,
and $E_{I}(x,y,t)$ is the $N\times N$ determinant of the same set of the columns in \eqref{eq:Et}, which is given by
\begin{equation}\label{eq:EI}
E_{I}=\prod_{k<l}(\kappa_{i_l}-\kappa_{i_k})\,E_{i_1}\cdots E_{i_N}=\prod_{k<l}(\kappa_{i_l}-\kappa_{i_k})
\exp\left(\xi_{i_1}+\cdots+\xi_{i_N}\right).
\end{equation}
The minor $\Delta_{I}(A)$ is also called the Pl\"ucker coordinate, and the $\tau$-function represents a
point of $\text{Gr}(N,M)$ in the sense of the Pl\"ucker embedding, $\text{Gr}(N,M)\hookrightarrow \mathbb{P}(\wedge^N\C^M): A\mapsto \{\Delta_I(A):I\in\binom{[M]}{N}\}$. It is shown in \cite{KW:13} that all the regular KP solitons can be obtained by choosing $A$ from the totally nonnegative Grassmannian, where $\Delta_J(A)\ge 0$ for each index set $J=\{j_1,\ldots,j_N\}$.

\subsection{The $M$-theta function as the $\tau$-function of the KP solitons}
Here we give the main theorem, which shows that the $\tau$-function \eqref{eq:tauPl} can be expressed by the $M$-theta function with further identifications
of $\{b_i,c_i\}$ on $\mathbb{P}$.

Let us start with the definition of the matroid for an irreducible matrix $A$ in the totally nonnegative Grassmannian $\text{Gr}(N,M)$, 
\[
\mathcal{M}(A)=\left\{I\in\binom{[M]}{N}: \Delta_I(A)>0\right\}.
\]
Let $I_0$ be the lexicographically minimum element of $\mathcal{M}(A)$. Then we have the decomposition,
\begin{equation}\label{eq:decompM}
\mathcal{M}(A)=\bigcup_{n=0}^N \mathcal{M}_n(A),
\end{equation}
where
\[
\mathcal{M}_n(A):=\left\{J\in\mathcal{M}(A):|J\cap I_0|=N-n\right\}.
\]
Note that $\mathcal{M}_0(A)=I_0$. 

Let $[i,j]$ denote the difference of the functions $\xi_j(x,y,t)=\kappa_jx+\kappa_j^2y+\kappa_j^3t$,
\begin{equation}\label{eq:ij}
[i,j]:=\xi_j-\xi_i,\quad\text{for}\quad i\in I,~j\in J,
\end{equation}
where $I,J\in\mathcal{M}(A)$. Notice that $(\alpha,\beta,\gamma)$ in $\xi_j-\xi_i=\alpha\,x+\beta\,y+\gamma\,t$ where 
\[
\alpha=\kappa_j-\kappa_i,\quad \beta=\kappa_j^2-\kappa_j^2,\quad \gamma=\kappa_j^3-\kappa_i^3,
\]
satisfies the dispersion relation of the KP equation \eqref{eq:dispersion}.
\begin{remark}
The $\tau$-function defined by $\tau=1+e^\phi$ with $\phi=[i,j]$ gives one-soliton solution,
\[
u(x,y,t)=2\partial_x^2\log\tau(x,y,t)=\frac{(\kappa_i-\kappa_j)^2}{2}\sech^2\frac{1}{2}\phi(x,y,t).
\]
\end{remark}
Let us define the difference $I_0-J$ for $J\in\mathcal{M}_1(A)$, in which each element for $J\in \mathcal{M}_1(A)$ is also denoted as $[i,j]$ for $i\in I_0\setminus J$ and
$j\in J\setminus I_0$. We then define the set
\[
P_1(A)=\{[i,j]:i\in I_0\setminus J,~j\in J\setminus I_0,~ J\in\mathcal{M}_1(A)\}
\]
Note that each element in $P_1(A)$ can be expressed as $[i_k,j_l]$ where 
\begin{itemize}
\item[(a)] $i_k\in I_0\setminus J$ is the $k$-th pivot of $A$,
\item[(b)]  $j_l\in J\setminus I_0$ is a nonzero element in the $k$-th row of $A$.
\end{itemize}
Let $j^{(k)}_l\in J\setminus I_0$ $(l=1,\ldots, n_k)$ be the nonzero elements in the $k$th row of $A$ with the order $j_l^{(k)}<j_{m}^{(k)}$ for $l<m$. Then we define
\begin{equation}\label{eq:phi}
\phi_{\tilde g_{k-1}+l}:=[i_k,j_l^{(k)}]=\xi_{j_l^{(k)}}-\xi_{i_k},\quad l=1,\ldots,n_k,
\end{equation}
where $\tilde g_{k}=n_1+\ldots+n_k$ with $\tilde g_0=0$. 
Let $\tilde g$ denote the total number of elements in $\mathcal{M}_1(A)$ given by
\begin{equation}\label{eq:Tg}
\tilde g:=\tilde g_N=|\mathcal{M}_1(A)|=n_1+\cdots+n_N,
\end{equation}
which is the total number of nonzero elements in $A$ except the pivots. Then $\tilde g$ gives the number of pairs $\{\kappa_{i_k},\kappa_{j_l}\}$ in the inverse image of singular points of $C$, $S=\{p_1,\ldots,p_{\tilde g}\}$, where we identify the points $\pi^{-1}(p_j)=\{b_j,c_j\}$ for $1\le j\le \tilde g$ as
\begin{equation}\label{eq:bc}
\{b_{\tilde g_{k-1}+l}=\kappa_{i_k},~ c_{\tilde g_{k-1}+l}=\kappa_{j_l^{(k)}}\},\quad \text{for}\quad 1\le l\le n_k~\,\text{and}~\, 1\le k\le N.
\end{equation}
One should note here that the indices of $b_i$'s are given by $I_0$, the pivot indices, and those of $c_j$'s are given by the non-pivot indices.

We also define the following sets for $2\le n\le N$. For $J=\{j_1,\ldots,j_N\}\in\mathcal{M}_n(A)$, 
\[
P_n(A):=\left\{\sum_{i_k\in I_0\setminus J}\sum_{j_l\in J\setminus I_0}[i_k,j_l]:[i_k,j_l]\in P_1(A),~ J\in\mathcal{M}_n(A)\right\}
\]
Note that the sets $P_n(A)$ are generated by the elements in $P_1(A)$, and each index in the sum appears only once.
The number of elements $[i_k,j_l]$ in each sum in $P_n(A)$ is $n$, and we denote the element as
\[
J-I_0=\sum_{i_k\in I_0\setminus J}\sum_{j_l\in J\setminus I_0}[i_k,j_l].
\]
Note in particular that the sum is not unique, e.g. $[1,4]+[2,3]=[1,3]+[2,4]$ (see the example below).
We then define the linear dependence of the elements in $P_1(A)$.
\begin{definition}
The linear independence and the dimension of $P_1(A)$.
\begin{itemize}
\item[(a)] $P_1(A)$ is said to be linearly dependent if an element $[i_k,j_l]$ in $P_1(A)$ is expressed by
\[
[i_k,j_l]=\sum\epsilon_{i,j}[i,j],\quad\text{with}~ \epsilon_{i,j}\in\{0,\pm1\}~\text{and}~\sum|\epsilon_{i,j}|\ge 2,
\]
where the sum is over all elements in $P_1(A)$.
\item[(b)] The dimension of $P_1(A)$ is defined by
\[
\text{dim}\, P_1(A)=\{\text{number of independent elements in}~P_1(A)\}
\]
Note that $\text{dim}\,P_1(A)\le \tilde g$ defined in \eqref{eq:Tg}.
\end{itemize}
\end{definition}
Then we have the following lemma.
\begin{lemma}\label{lem:dimP1}
For $A\in \text{Gr}(N,M)$ (irreducible totally nonnegative element), we have
\[
N\le \text{dim}\,P_1(A)\,\le M-1.
\]
\end{lemma}
\begin{proof}
Since $A$ is irreducible, there is at least one $[i_k,j]$ for each row. This implies $N\le \text{dim}\,P_1(A)$.
The maximum dimension is given by the case when $A$ is in the top cell (largest dimensional cell in $\text{Gr}(N,M)$.
In this case, we have $|\mathcal{M}_1(A)|=N(M-N)$. Each dependent relation is corresponding to the $2\times 2$ matrix,
\[
\begin{pmatrix}
a_{k,N+i}& a_{k,N+i+1}\\
a_{k+1,N+i}&a_{k+1,N+i+1}
\end{pmatrix},\qquad \text{for}\quad 1\le i\le M-N-1~\text{and}~~1\le k\le N-1
\]
that is, $[k,N+i]+[k+1,N+i+1]=[k,N+i+1]+[k+1,N+i]$, which give the total $(N-1)(M-N-1)$ independent relations.
Then we have 
\[
\text{dim}\,P_1(A)=N(M-N)-(N-1)(M-N-1)=M-1,
\]
which gives the maximum dimension.
\end{proof}

\begin{example}
Consider the element of the top positive Grassmannian cell $\text{Gr}(2,4)$, which is
\[
A=\begin{pmatrix}
1 &0 & -a & -b\\
0&1&c & d
\end{pmatrix}
\]
where $a,b,c,d>0$ and $cb-ad>0$. Then we have $\mathcal{M}(A)=\mathcal{M}_0(A)\cup\mathcal{M}_1(A)\cup\mathcal{M}_2(A)$ where
\begin{align*}
\mathcal{M}_0(A)=\{(1,2)\},\quad
\mathcal{M}_1(A)=\{(1,3),~(1,4),~(2,3),~(2,4)\},\quad
\mathcal{M}_2(A)=\{(3,4)\}.
\end{align*}
This gives
\begin{align*}
P_1(A)&=\{[1,3],~[1,4],~[2,3],~[2,4]\},\\
P_2(A)&=\{[1,3]+[2,4]=[1,4]+[2,3]\},
\end{align*}
which implies that the set $P_1(A)$ is linearly dependent and
$\text{dim}\,P_1(A)=3<\tilde g=|\mathcal{M}_1(A)|=4$
\end{example}

We write the $\tau$-function \eqref{eq:tauPl} in the following form,
\begin{align}\label{eq:tauNormal}
\tau(x,y,t)&=\sum_{n=0}^N\sum_{J\in\mathcal{M}_n(A)}\Delta_J(A)E_J(x,y,t)\\
&=\Delta_{I_0}(A)E_{I_0}(x,y,t)\left(1+\sum_{n=1}^N\sum_{J\in\mathcal{M}_n(A)}\frac{\Delta_J(A)E_J}{\Delta_{I_0}(A)E_{I_0}}\right)\nonumber
\end{align}
Since the solution is given by the second derivative of $\ln\tau$, one can take the $\tau$-function in the following form,
\begin{equation}\label{eq:grammian}
\tau(x,y,t)=1+\sum_{n=1}^N\sum_{J\in\mathcal{M}_n(A)}\Delta_J(A)\frac{E_J}{E_{I_0}}.
\end{equation}
where we have taken $\Delta_{I_0}(A)=1$ for the pivot set $I_0$.
For each $i_k\in I_0$ and $j^{(k)}_l\in J$ of $J\in\mathcal{M}_1(A)$ (see \eqref{eq:phi} for the notation), we first note from \eqref{eq:EI} that
\[
\frac{E_J}{E_{I_0}}=\frac{\prod_{m<n}(\kappa_{j_n}-\kappa_{j_m})}{\prod_{m<n}(\kappa_{i_n}-\kappa_{i_m})}\exp({\xi_{j^{(k)}_l}-\xi_{i_k}})
=\frac{\prod_{m\ne k}|\kappa_{i_m}-\kappa_{j^{(k)}_l}|}{\prod_{m\ne k}|\kappa_{i_m}-\kappa_{i_k}|}\exp({\phi_{\tilde g_{k-1}+l}}).
\]
Then we define a constant phase parameter $\phi_j^{(0)}$ for $j=1,\ldots, \tilde g$,
\begin{equation*}\label{eq:phi0}
\exp(\phi_{\tilde g_{k-1}+l}^0):=\Delta_J(A)\frac{\prod_{m\ne k}|\kappa_{i_m}-\kappa_{j^{(k)}_l}|}{\prod_{m\ne k}|\kappa_{i_m}-\kappa_{i_k}|},\quad 1\le l\le n_k.
\end{equation*}
Note here that each nonzero entry in the matrix $A$ is represented in the phase function $\phi_m^0$ for $1\le m\le \tilde g:=|\mathcal{M}_1(A)|$, and we have $\Delta_J(A)=|a_{k,j_l^{(k)}}|$, i.e.
\begin{equation}\label{eq:aij}
|a_{k,j_l^{(k)}}|=\frac{\prod_{m\ne k}|\kappa_{i_m}-\kappa_{i_k}|}{\prod_{m\ne k}|\kappa_{i_m}-\kappa_{j_l^{(k)}}|}\exp(\phi^0_{\tilde g_{k-1}+l}).
\end{equation}
Then the first term in the sum in \eqref{eq:grammian} is given by
\[
\sum_{J\in\mathcal{M}_1(A)}\Delta_J(A)\frac{E_J}{E_{I_0}}=\sum_{k=1}^N\sum_{l=1}^{n_k}\exp(\tilde\phi_{\tilde g_{k-1}+l})=\sum_{m=1}^{\tilde g}e^{\tilde\phi_m},
\]
where $\tilde\phi_m:=\phi_m+\phi^0_m$.

For the sum over $J\in\mathcal{M}_2(A)$,
let $\{i_k,i_l\}\in I_0\setminus J$ and $\{j_m^{(k)},j_n^{(l)}\}\in J\setminus I_0$. Then the following lemma is immediate
\begin{lemma}\label{lem:EJ2}
\[
\frac{E_J}{E_{I_0}}=\frac{\prod_{s\ne k,l}|\kappa_{i_s}-\kappa_{j_m^{(k)}}||\kappa_{i_s}-\kappa_{j_n^{(l)}}||\kappa_{j_m^{(k)}}-\kappa_{j_n^{(l)}}|}{\prod_{s\ne k,l}|\kappa_{i_s}-\kappa_{i_k}||\kappa_{i_s}-\kappa_{i_l}||\kappa_{i_k}-\kappa_{i_l}|}\exp({\phi_{\tilde g_{k-1}+m}+\phi_{\tilde g_{l-1}+n}}),
\]
\end{lemma}
Now let $a_{k,j_m^{(k)}}$ and $a_{l,j_n^{(l)}}$ be the entries of $A$ at the locations $(k,j_m^{(k)})$ and $(l,j_n^{(l)})$, respectively. 

\begin{proposition}\label{prop:Aij}
\begin{align*}
|a_{k,j_m^{(k)}}a_{l,j_n^{(l)}}|\frac{E_J}{E_{I_0}}&=\frac{|\kappa_{i_k}-\kappa_{i_l}||\kappa_{j_m^{(k)}}-\kappa_{j_n^{(l)}}|}{|\kappa_{i_k}-\kappa_{j_n^{(l)}}||\kappa_{j_m^{(k)}}-\kappa_{i_l}|}\exp(\tilde\phi_{\tilde g_{k-1}+m}+\tilde\phi_{\tilde g_{l-1}+n}),
\end{align*}
where $\tilde\phi_m=\phi_m+\phi_m^0$.
\end{proposition}
\begin{proof}
Let $J_k:=(J\setminus\{j_m^{(k)}\})\cup\{i_k\}$ and $J_l:=(J\setminus\{j_n^{(l)}\})\cup\{i_l\}$. Note that $J_k,J_l\in\mathcal{M}_1(A)$, and we have
\[
\Delta_{J_k}(A)=|a_{k,j_m^{(k)}}|,\quad\text{and}\quad \Delta_{J_l}=|a_{l,j_n^{(l)}}|.
\]
Then using \eqref{eq:aij}, we have the formula.
\end{proof}
From \eqref{eq:bc}, the coefficient of the exponential function gives
\[
\frac{(\kappa_{i_k}-\kappa_{i_l})(\kappa_{j_m^{(k)}}-\kappa_{j_n^{(l)}})}{(\kappa_{i_k}-\kappa_{j_n^{(l)}})(\kappa_{j_m^{(k)}}-\kappa_{i_l})}=\exp\left(2\pi i\,\tilde\Omega_{g_{k-1}+m,g_{l-1}+n}\right)
\]
One should note that for $\{j_m^{(k)},j_n^{(l)}\}\in J\setminus I_0$, we have
\[
\Delta_J(A)=\Big|a_{k,j_m^{(k)}}a_{l,j_n^{(l)}}-a_{k,j_{n'}^{(k)}}a_{l,j_{m'}^{(l)}}\Big|,
\]
where $j_{n'}^{(k)}=j_n^{(l)}$ and $j_{m'}^{(l)}=j_m^{(k)}$.  If $a_{k,j_{n'}^{(k)}}a_{l,j_{m'}^{(l)}}\ne0$, we have the phases
\[
\phi_{\tilde g_{k-1}+n'}=[i_k,j_{n'}^{(k)}],\quad\text{and}\quad \phi_{\tilde g_{l-1}+m'}=[i_l,j_{m'}^{(l)}].
\]
Note here that we have a resonant condition,
\[
\phi_{\tilde g_{k-1}+n'}+\phi_{\tilde g_{l-1}+m'}=\phi_{\tilde g_{k-1}+m}+\phi_{\tilde g_{l-1}+n}.
\]
Then we have
\begin{align*}
\Delta_J(A)\frac{E_J}{E_{I_0}}=&\pm(a_{k,j_m^{(k)}}a_{l,j_n^{(l)}}-a_{k,j_{n'}^{(k)}}a_{l,j_{m'}^{(l)}})\frac{E_J}{E_{I_0}}\\
=&\pm\left(\frac{(\kappa_{i_k}-\kappa_{i_l})(\kappa_{j_m^{(k)}}-\kappa_{j_n^{(l)}})}{(\kappa_{i_k}-\kappa_{j_n^{(l)}})(\kappa_{j_m^{(k)}}-\kappa_{i_l})}e^{\tilde\phi_{\tilde g_{k-1}+m}+\tilde\phi_{\tilde g_{l-1}+n}}\right.\\
&\qquad \left.-\frac{(\kappa_{i_k}-\kappa_{i_l})(\kappa_{j_{n'}^{(k)}}-\kappa_{j_{m'}^{(l)}})}{(\kappa_{i_k}-\kappa_{j_{m'}^{(l)}})(\kappa_{j_{n'}^{(k)}}-\kappa_{i_l})}e^{\tilde\phi_{\tilde g_{k-1}+n'}+\tilde\phi_{\tilde g_{l-1}+m'}}\right)
\end{align*}

Repeating the similar computations, we have the formula for $J\in\mathcal{M}_n(A)$.
\begin{proposition}
Let $\{i_k, i_l,\ldots,i_m\}\in I_0\setminus J$ and $\{j_r^{(k)}, j_s^{(l)}, \ldots, j_t^{(m)}\}\in J\setminus I_0$. Then we have
\[
\left| a_{k,j_r^{(k)}}\,a_{l,j_s^{(l)}}\cdots a_{m,j_t^{(m)}}\right|\frac{E_J}{E_{I_0}}=\left|C_{\tilde g_{k-1}+r,\, \tilde g_{l-1}+s,\, \ldots, \,\tilde g_{m-1}+t}\right|e^{\tilde \phi_{\tilde g_{k-1}+r}+\tilde\phi_{\tilde g_{l-1}+s}+\cdots+\tilde \phi_{\tilde g_{m-1}+t}},
\]
where $C_{k,l,\ldots,m}$ is defined by
\[
C_{k,l,\ldots,m}=\prod_{a<b\in\{k,l,\ldots,m\}}C_{a,b}=\exp 2\pi i\left(\sum_{a<b\in\{k,l,\ldots,m\}}\tilde\Omega_{a,b}\right),
\]
with
\[
C_{\tilde g_{k-1}+r,\,\tilde g_{l-1}+s}:=\frac{(\kappa_{i_k}-\kappa_{i_l})(\kappa_{j_r^{(k)}}-\kappa_{j_s^{(l)}})}{(\kappa_{i_k}-\kappa_{j_s^{(l)}})(\kappa_{j_r^{(k)}}-\kappa_{i_l})}=\exp 2\pi i\,\tilde\Omega_{\tilde g_{k-1}+r,\tilde g_{l-1}+s}.
\]
\end{proposition}

We now state the main theorem as the summary of the results above.
\begin{theorem}\label{thm:main}
For $A\in\text{Gr}(N,M)$ ($A$ is in the RREF), let $\tilde g$ be the total number of nonzero elements in $A$ except the pivot one's.
Let $n_k$ denote the total number of nonzero elements in the $k$-th row, and define
\[
\tilde g_k=n_1+\cdots+n_k,\quad (1\le k\le N),\qquad \text{and}\quad \tilde g=\tilde g_N.
\]
Also let $j^{(k)}_m$ denote the non-pivot column index of a nonzero element in the $k$-th row of $A$ (note $1\le m\le n_k$),
and assume that $j_m^{(k)}<j_l^{(k)}$ if $m<l$.

Then the $\tau$-function $\tau_A(x,y,t)$ in \eqref{eq:tauNormal} can be expressed by the (degenerated) $M$-theta function
$\tilde\vartheta_{\tilde g}^{(\tilde g)}(\z;\tilde\Omega)$,
\[
\tau_A(x,y,t)/E_{I_0}(x,y,t)=\tilde\vartheta^{(\tilde g)}_{\tilde g}(\z;\tilde\Omega)=\sum_{m\in\{0,1\}}^{\tilde g}\exp 2\pi i\left(\sum_{i<j}m_im_j\tilde\Omega_{i,j}+\sum_{j=1}^{\tilde g}m_jz_j\right),
\]
where $E_{I_0}$ is the $N\times N$ determinant in \eqref{eq:EI} with the lexicographically minimal set $I=I_0$, $\tilde g=\tilde g_N$ and  $2\pi i \z=\tilde\phi(x,y,t)$ with $\tilde\phi_n=\phi_n+\phi^0_n$ for $1\le n\le \tilde g$. We then have for $1\le k\le N$
\begin{align*}
&\phi_{\tilde g_{k-1}+m}=[i_k,j^{(k)}_m]=(\kappa_{j_m^{(k)}}-\kappa_{i_k})x+(\kappa_{j_m^{(k)}}^2-\kappa^2_{i_k})y+(\kappa^3_{j_m^{(k)}}-\kappa^3_{i_k})t,\\
&\exp(\phi^0_{\tilde g_{k-1}+m})=| a_{k,j_m^{(k)}}|\frac{\prod_{l\ne k}|\kappa_{i_l}-\kappa_{j^{(k)}_m}|}{\prod_{l\ne k}|\kappa_{i_l}-\kappa_{i_k}|},\\
&\exp\left(2\pi i\,\tilde\Omega_{g_{k-1}+m,g_{l-1}+n}\right)=\frac{(\kappa_{i_k}-\kappa_{i_l})(\kappa_{j_m^{(k)}}-\kappa_{j_n^{(l)}})}{(\kappa_{i_k}-\kappa_{j_n^{(l)}})(\kappa_{j_m^{(k)}}-\kappa_{i_l})}.
\end{align*}
Here $a_{k,j^{(k)}_m}$ is the element in $A$ with the corresponding index, and assume $\tilde g_0=0$.
\end{theorem}

\begin{remark}
Having the variables, $\{\tilde\phi_j, \tilde\Omega_{j,k}:1\le j<k\le \tilde g\}$ from the matrix $A\in\text{Gr}(N,M)$ and the parameters $(\kappa_1,\ldots,\kappa_{M})$, the $\tau$-function can be given in the form with the vertex operators \eqref{eq:Ltheta}. This formula has been also obtained by Nakayashiki \cite{N:23}, where he has constructed the $\tau$-function by
taking several limits of the parameters of the vertex operators for the Hirota type solutions.
\end{remark}

\begin{remark}
From Remark \ref{rem:Grammian}, Theorem \ref{thm:main} implies that the Wronskian form of the $\tau$-function can be also
written in the Grammian form which is the $M$-theta function \cite{CLM:10}.
\end{remark}

\begin{remark}\label{rem:BKP}
In the Pfaffian formula for the Prym $H$-theta function, we take
\[
\xi_j=\sum_{k=\text{odd}}q_j^kt_k=\left\{\begin{array}{lll}
\sum_{k=\text{odd}}b_n^kt_k,&\quad\text{if}~j=2n-1,\\
\sum_{k=\text{odd}}c_n^kt_k,&\quad\text{if}~j=2n.
\end{array}\right.\quad\text{for}\quad 1\le j\le 2g.
\]
(The time variables can include $t_{-1}$ \cite{H:04}.)
Then the element of the skew symmetric matrix $Q$ in \eqref{eq:Btau}, $\tilde\vartheta^B_g(z;\tilde\Omega^B)=\text{Pf}(Q)$ can be expressed by
\[
Q_{j,k}=c_{j,k}+\int^x D_1 e^{\xi_j}\cdot e^{\xi_k}\,dx,
\]
where $D_1$ is the Hirota derivative with respect to $t_1$, and the $\tau$-function \eqref{eq:Btau} gives a solution of the BKP equation \cite{H:04}, 
\[
(D_1^6-5D_1^3D_3-5D_3^2+9D_1D_5)\tau\cdot\tau=0.
\]
In a future communication, we plan to prove a classification theorem for the regular BKP solitons similar to Theorem \ref{thm:main}.
\end{remark}

\section{Examples}\label{sec:Example}
\subsection{Hirota type soliton solutions}
First note that the $M$-theta function \eqref{eq:Ltheta} gives
the Hirota type $g$-soliton solution \cite{H:04} by choosing the following variables with $\{b_j,c_j\}=\pi^{-1}(p_j)$,
\[
\tilde\phi_j=(c_j-b_j)x+(c_j^2-b_j^2)y+(c_j^3-b_j^3)t+\phi_j^0,\quad j=1,\ldots,g,
\]
where $b_j,c_j$ and $\phi_j^0$ are arbitrary constants. The coefficients $C_{k,l}:=e^{2\pi i\tilde\Omega_{k,l}}$ are
\[
C_{k,l}=\frac{(b_k-b_l)(c_k-c_l)}{(b_k-c_l)(c_k-b_l)}.
\]
For the regular real solitons, we assume all the parameters $\{b_j, c_j, \phi_j^0\}$ be real with $b_j<c_j$ and $C_{k,l}>0$, which leads to the following conditions,
\begin{itemize}
\item[(a)] $b_k<b_l<c_l<c_k$, or
\item[(b)] $b_k<c_k<b_l<c_l$, 
\end{itemize}
(the case (a) is called $P$-type, and the case (b) is called $O$-type in \cite{K:17}). In terms of the matrix $A\in \text{Gr}(g,2g)$,
the $b_j$ corresponds to a pivot position ($b_1<\cdots<b_g$) and the $c_j$ to a non-pivot position. For example, $A_{(a)}$ for the case (a) and $A_{(b)}$ for the case (b) in $\text{Gr}(2,4)$ are given by
\[
A_{(a)}=\begin{pmatrix}
1 &0 & 0 & -a\\
0&1& b& 0
\end{pmatrix}\quad\text{and}\quad
A_{(b)}=\begin{pmatrix}
1&a&0&0\\
0&0&1&b
\end{pmatrix},
\]
that is, for $A_{(a)}$, $b_1=\kappa_1, c_1=\kappa_4, b_2=\kappa_2, c_2=\kappa_3$, and for $A_{(b)}$, $b_1=\kappa_1, c_1=\kappa_2, b_2=\kappa_3, c_2=\kappa_4$. In general, the matrix $A$ consists of both (a) and (b) types. For example, we have
\[
A=\begin{pmatrix}
1&0&0&0&0&-a\\
0&1&b&0&0&0\\
0&0&0&1&c&0
\end{pmatrix}.
\]
The $\tau$-function in the form of the $M$-theta function is given as follows. First note that $|\mathcal{M}_1(A)|=3$, i.e.
there are 3 singular points, $S=\{p_1,p_2,p_3\}$, and we have
\[
\phi_1=[1,6]=\xi_6-\xi_1,\quad \phi_2=[2,3]=\xi_3-\xi_2,\quad \phi_3=[4,5]=\xi_5-\xi_4,
\]
which gives $\pi^{-1}(p_j)=\{b_j,c_j\}$, the inverse image of the normalization,
\[
b_1=\kappa_1,\quad c_1=\kappa_6,\qquad b_2=\kappa_2,\quad c_2=\kappa_3,\qquad b_3=\kappa_4,\quad c_3=\kappa_5.
\]
For the free parameters $\{a,b,c\}$ in the matrix $A$, we introduce $\phi^0_j$ as
\begin{align*}
&e^{\phi_1^0}=a\frac{(\kappa_2-\kappa_4)(\kappa_2-\kappa_6)(\kappa_4-\kappa_6)}{(\kappa_1-\kappa_2)(\kappa_1-\kappa_4)(\kappa_2-\kappa_4)},\\
&e^{\phi_2^0}=b\frac{(\kappa_1-\kappa_3)(\kappa_1-\kappa_4)(\kappa_3-\kappa_4)}{(\kappa_1-\kappa_2)(\kappa_1-\kappa_4)(\kappa_2-\kappa_4)},\\
&e^{\phi_3^0}=c\frac{(\kappa_1-\kappa_2)(\kappa_1-\kappa_5)(\kappa_2-\kappa_5)}{(\kappa_1-\kappa_2)(\kappa_1-\kappa_4)(\kappa_2-\kappa_4)},
\end{align*}
Then, the $M$-theta function is given by
\begin{align*}
\tilde\vartheta_3^{(3)}&=1+\sum_{j=1}^3e^{\tilde\phi_j}+\sum_{1\le j<k\le 3}C_{j,k}e^{\tilde\phi_j+\tilde\phi_k}+C_{1,2}C_{1,3}C_{2,3}e^{\tilde\phi_1+\tilde\phi_2+\tilde\phi_3},
\end{align*}
where $\tilde\phi_j=\phi_j+\phi_j^0$ with $C_{k,l}$ given by \eqref{eq:Omega}.

In general, the soliton solution of Hirota type consists of $g$ line solitons, and if $b_k+c_k\ne b_l+c_l$, the corresponding solitons $\phi_k=[b_k,c_k]$ and $\phi_l=[b_l,c_l]$ intersect with phase shift computed by $C_{k,l}$. Depending of the types (a) or (b) above, we have $C_{k,l}<1$ (negative phase shift) for (a), and $C_{k,l}>1$ (positive phase shift) for (b). 

\begin{remark}\label{rem:gas}
We note that the dual graph of the soliton graph (defined in \cite{K:17}) gives a rhombus tiling with $\binom{g}{2}$ rhombuses, in which each rhombus is labeled by the intersecting two solitons $\phi_k$ and $\phi_l$. Then coloring each rhombus according to the types either $C_{k,l}>1$ or $C_{k,l}<1$, we have a colored rhombus tiling of $2g$-gon.
The total number of different colored tilings is given by a Catalan number \cite{S:97, CK:08},
\[
C_g=\frac{1}{g+1}\binom{2g}{g}.
\]
It is then interesting to notice that a KP soliton with a set of random $C_{k,l}$'s (corresponding to the positive and negative phase shifts) may provide a two-dimensional extension of the KdV soliton gas (see \cite{El:23} and the references therein).
\end{remark}

\subsection{A degenerate example}
Consider an $2\times 4$ matrix $A$ in the positive $\text{Gr}(2,4)$, which is given by
\[
A=\begin{pmatrix}
1 & 0 & -a & -b\\
0& 1 & c & d
\end{pmatrix}
\]
where $a, b, c, d$ are positive constants with $cb-ad>0$.  We introduce the new variables,
\[
\phi_1:=[1,3]=\xi_3-\xi_1,\quad \phi_2:=[1,4]=\xi_4-\xi_1,\quad \phi_3:=[2,3]=\xi_3-\xi_2,\quad \phi_4:=[2,4]=\xi_4-\xi_2.
\]
This implies that the pairs of points $\pi^{-1}(p_j)=\{b_j,c_j\}$ of the normalization for $\tilde g=4$ are given by
\begin{equation}\label{eq:Gr24bc}
\begin{matrix} 
b_1=\kappa_1,\quad & b_2=\kappa_1,\quad & b_3=\kappa_2,\quad & b_4=\kappa_2,\\
c_1=\kappa_3, & c_2=\kappa_4, &\kappa_3=\kappa_3, & c_4=\kappa_4.
\end{matrix}
\end{equation}
Note here that some of the points are also identified, and it leads to further degenerations of $\Omega_{k,l}\to +i\infty$ (see below).
Also note that these phases are linearly dependent, that is,
\[
\phi_4=[2,4]=[2,3]+[3,1]+[1,4]=\phi_3+\phi_2-\phi_1,
\]
which implies that $\text{dim}\,P_1(A)=3$.
We now introduce the phase constants $e^{\phi_j^0}$ to represent the nonzero parameters in $A$ except the pivot ones,
\begin{align*}
e^{\phi_1^0}:=&a\frac{\kappa_3-\kappa_2}{\kappa_2-\kappa_1},\quad e^{\phi_2^0}:=b\frac{\kappa_4-\kappa_2}{\kappa_2-\kappa_1},\quad
e^{\phi_3^0}:=c\frac{\kappa_3-\kappa_1}{\kappa_2-\kappa_1},\quad e^{\phi_4^0}:=d\frac{\kappa_4-\kappa_1}{\kappa_2-\kappa_1}.
\end{align*}
 The phase variables are then given by
$ 2\pi iz_k=\tilde\phi_k=\phi_k+\phi_k^0,\quad \text{for}~\, k=1,\ldots, 4$, 
and the coefficients $C_{k,l}:=e^{2\pi i\tilde\Omega_{k,l}}=\frac{(b_k-b_l)(c_k-c_l)}{(b_k-c_l)(c_k-b_l)}$ become \eqref{eq:Gr24bc},
\[
C_{1,2}=C_{1,3}=C_{2,4}=C_{3,4}=0,\quad C_{2,3}=\frac{(\kappa_1-\kappa_2)(\kappa_3-\kappa_4)}{(\kappa_1-\kappa_3)(\kappa_2-\kappa_4)},\quad
C_{1,4}=-\frac{(\kappa_1-\kappa_2)(\kappa_3-\kappa_4)}{(\kappa_1-\kappa_4)(\kappa_2-\kappa_3)}.
\]
Then, the $\tau$-function for the matrix $A$ can be expressed in the form of $\tilde\vartheta_4^{(4)}$, i.e. $\tau_A=E_{1,2}\tilde\vartheta_4^{(4)}$ with
\begin{align}\label{eq:tauGr24A}
\tilde\vartheta_4^{(4)}=&1+e^{\tilde\phi_1}+e^{\tilde\phi_2}+e^{\tilde\phi_3}+e^{\tilde\phi_4}+\\
&+\frac{(\kappa_1-\kappa_2)(\kappa_3-\kappa_4)}{(\kappa_1-\kappa_3)(\kappa_2-\kappa_4)}e^{\tilde\phi_2+\tilde\phi_3}
-\frac{(\kappa_1-\kappa_2)(\kappa_3-\kappa_4)}{(\kappa_1-\kappa_4)(\kappa_2-\kappa_3)}e^{\tilde\phi_1+\tilde\phi_4}.\nonumber
\end{align}
where $E_{1,2}=(\kappa_2-\kappa_1)e^{\xi_1+\xi_2}$.
Here the additional degenerations in $\Omega_{k,l}$ are
\[
\tilde\Omega_{1,2}=\tilde\Omega_{1,3}=\tilde\Omega_{2,4}=\tilde\Omega_{3,4}=+i\infty.
\]
This implies that the interactions between these waves give an infinite phase shift, which represents a resonant interaction.
Note here that $C_{1.4}=e^{2\pi i\tilde\Omega_{1,4}}<0$, i.e. the real part of $\tilde\Omega_{1,4}$ is $\text{Re}(\tilde\Omega_{1,4})=\frac{1}{2}$. Then the positivity of $\Delta_{3,4}(A)=cb-ad$ is due to the values of $\phi^0_1+\phi_4^0$
and $\phi^0_2+\phi^0_3$.

\begin{remark}
The dimension of $P_1(A)$ represents of the number of independent phases. In this example of $\text{Gr}(2,4)$, the dimension is three, but the KP soliton should be embedded into $\tilde g=4$ $M$-theta function. This implies that the set of singular points on the curve is $S=\{p_1,\ldots,p_4\}$. We will discuss
the details of the singular structures of the curve associated to KP solitons in a future communication.
\end{remark}

\begin{remark}
As is also mentioned in \cite{AFMS:23}, the KP solitons can be considered as tropical limits of the theta function, and
the tropical Riemann period matrix defines a Delaunary polytope (dual graph is the soliton graph \cite{K:18}). The dimension of the polytope should be considered as dim$P_1(A)$, not $\tilde g=|\mathcal{M}_1(A)|$. For the example above of $\text{Gr}(2,4)$, we have $\text{dim}P_1(A)=3$, and the corresponding polytope is octahedron. In \cite{AFMS:23}, it is shown that 
there are total of 5 distinct polytopes such as tetrahedra, octahedra, square pyramid, triangular prism and cube. They are easily
identified as the (dual) soliton graphs in $(x,y,t)$ space. The higher dimensional examples are obtained by the KP hierarchy (see e.g. \cite{KW:15}). We will discuss the details in a future communication.
\end{remark}

\section{KP solitons on a quasi-periodic background}\label{sec:SolitonQ}
In the limit \eqref{eq:limits}, if we consider only for some $\Omega_{j,j}\to+i\infty$, we get solitons on a quasi-periodic background. Let us consider the limit with $\Omega_{i,i}\to+i\infty$ for $1\le i\le n<g$.
Then having the shift $z_i\to z_i-\frac{1}{2}\Omega_{i.i}$ for $1\le i\le n$, we have
\begin{align}\label{eq:Ptheta}
\vartheta_g(\z;\Omega)~\to~\,~\tilde\vartheta^{(n)}_g(\z;\tilde\Omega)&:=
\vartheta_{g-n}(\z^{(n)};\tilde\Omega^{(n)})+
\sum_{j=1}^ne^{2\pi iz_j}\vartheta_{g-n}(\z^{(n)}+\tilde\Omega_j\,;\,\tilde\Omega^{(n)})+\\\nonumber
& +\sum_{1\le j\le k\le n}e^{2\pi i(\tilde\Omega_{j,k}+z_j+z_k)}\vartheta_{g-n}(\z^{(n)}+\tilde\Omega_j+\tilde\Omega_k\,;\,\tilde\Omega^{(n)})+\cdots\\\nonumber
&\quad\cdots+e^{2\pi i(\sum_{1\le j<k\le n}\tilde\Omega_{j,k}+\sum_{j=1}^nz_j)}\vartheta_{g-n}(\z^{(n)}+\sum_{j=1}^n\tilde\Omega_j\,;\,\tilde\Omega^{(n)}),
\end{align}
where we have defined 
\[
\z^{(n)}=(z_{n+1},z_{n+2},\cdots,z_g),\qquad \z^{(n)}+\tilde\Omega_j=(z_{n+1}+\tilde\Omega_{j,n+1},\ldots, z_g+\tilde\Omega_{j,g}).
\]
and $\tilde\Omega^{(n)}$ is the period matrix given by the bottom $(g-n)\times(g-n)$ part of the period matrix $\tilde\Omega$,
i.e. $\tilde\Omega^{(n)}=(\tilde\Omega_{k,l})_{n+1\le k,l\le g}$. The $\vartheta_{g-n}(\z^{(n)};\tilde\Omega^{(n)})$ is the theta
function of genus $g-n$ with the period matrix $\tilde\Omega^{(n)}$.
One can express the function $\tilde\vartheta^{(n)}_g$ using the vertex operator \eqref{eq:VO},
\begin{equation}\label{eq:VOPeriodic}
\tilde\vartheta^{(n)}_g(\z;\tilde\Omega)=\prod_{j=1}^n(1+\hat V_j[\tilde\Omega])\cdot\vartheta_{g-n}(\z^{(n)};\tilde\Omega^{(n)}).
\end{equation}
(Similar formula has been presented in Nakayashiki's lecture from different point of view \cite{N:23}.)
From Lemma \ref{lem:VO}, we see that the coefficients of the $\vartheta_{g-n}(\z^{(n)};\tilde\Omega^{(n)})$ depends only on the variables
$z_j$ for $1\le j\le n$ and the matrix $\tilde\Omega_{i,j}$ for $1\le i<j\le n$. Then considering a matrix $A\in\text{Gr}(N,M)$ with
$|\mathcal{M}_1(A)|=n$, one can construct $\tilde\Omega=(\tilde\Omega_{i,j})_{1\le i,j\le n}$ for a KP soliton following Theorem \ref{thm:main}. We then expect to have the KP soliton on quasi-periodic background given by the theta function $\vartheta_{g-n}(z^{(n)};\tilde\Omega)$.

The wave pattern of the KP soliton has a web-like structure, which divides several connected regions of the plane (see e.g. \cite{K:17}). Each connected region represents the dominant exponential in the $\tau$-function. Then the solution generated by \eqref{eq:VOPeriodic} is expected to have the similar structure of the KP soliton where each connected region presents a part of  quasi-periodic solution generated by $\vartheta_{g-n}$. And each soliton provides the phase shift to the adjacent quasi-periodic solutions. For example, consider the case with $g=2$ and $n=1$, and take the limit $\Omega_{1,1}\to+i\infty$. We have
\[
\vartheta_2^{(1)}(z_1,z_2,\tilde\Omega)=(1+\hat V_1)\vartheta_1(z_2;\tilde\Omega_{2,2})=\vartheta_1(z_2;\tilde\Omega_{2,2})+e^{2\pi iz_1}\vartheta_1(z_2+\tilde\Omega_{1,2};\tilde\Omega_{2,2}).
\]
This implies that in the region where $2\pi i z_1=\phi_1\ll0$, we have the periodic solution generated by
$\tau(x,y,t)=\vartheta_1(ux+vy+wt;\tilde\Omega_{2,2})$ (i.e. we take $z_2=ux+vy+wt$), and in the region where $2\pi iz_1=\tilde\phi_1\gg 0$,
we have the same periodic solution with the phase shift given by $2\pi i\tilde\Omega_{1,2}$. That is, one KP soliton interacts with the periodic background with phase shift. In his study of degeneration of hyperelliptic curve \cite{N:20}, Nakayashiki has obtained a similar formula \cite{N2:23}. Also in \cite{HMP:23, BJT:22}, the authors obtained a similar result for the KdV case.
We will discuss the general structure of the KP solitons on quasi-periodic background in a future communication.

\bigskip
\noindent
{\bf Acknowledgements.}
The author would like to thank Atsushi Nakayashiki for valuable discussions, and also received 
 benefits from the Nakayashiki's series of lectures \cite{N:23} (see also \cite{N2:23}).  He would also like to thank Atsushi Nakayashiki for informing him the references \cite{I:22, Ka:11}, and Sarbarish Chakravarty for critical reading of the manuscript. He appreciates a research fund from Shandong University of Science and Technology. He would also like to thank the referees for their useful and helpful suggestions.
 
 \bigskip
 \noindent
 {\bf Declarations.}
 \begin{itemize}
\item Data sharing not applicable to this article as no datasets were generated or analyzed during the current study.\\
 \item The author declares no conflicts of interest associated with this manuscript.
 \end{itemize}

\raggedright


\end{document}